\RequirePackage{fix-cm}

\documentclass{article}

\usepackage{graphicx}
\usepackage[lofdepth, lotdepth]{subfig}%[lofdepth, lotdepth]
\usepackage{listings}
\usepackage{multirow}
\usepackage{amsthm}

\usepackage{amsmath}

\newtheorem{proposition}{Proposition}[section]

\newtheorem{definition}{Definition}[section]

\newtheorem{corollary}{Corollary}[section]

\title{Objective Bayesian Analysis for the Lomax Distribution}

\author{Paulo Ferreira\thanks{Corresponding author. Email:
    paulohenri@ufba.br}, Jhon Gonzales, Vera Tomazella, \\ Ricardo Ehlers,
  Francisco Louzada, Eveliny Silva
}

\date{February 2016}

\begin{document}

\maketitle

\begin{abstract}
In this paper we propose to make Bayesian inferences for the
parameters of the Lomax distribution using non-informative priors, 
namely the Jeffreys prior and the reference prior. We assess Bayesian
estimation through a Monte Carlo study with 500 simulated data 
sets. To evaluate the possible impact of prior specification on
estimation, two criteria were considered: the bias and square root of
the mean square error. The developed procedures are illustrated on a
real data set.\\

\noindent Keywords: Bayesian inference, Jeffreys prior, Lomax
distribution, reference prior.

\end{abstract}

\section{Introduction} \label{intro}
The Lomax distribution \cite{lomax54}, also known as the Pareto Type II distribution (or simply Pareto II), is a heavy-tail probability distribution often used in business, economics and actuarial modeling. It is essentially a Pareto distribution that has been shifted so that its support begins at zero \cite{van2009}. The Lomax distribution has been applied in a variety of contexts ranging from modeling the survival times of patients after a heart transplant \cite{bain92} to the sizes of computer files on servers \cite{holland2006}. Some authors, such as \cite{bryson74}, suggest the use of this distribution as an alternative to the exponential distribution when data are heavy-tailed. 

The main objective of this paper is to make Bayesian inferences for the parameters of the Lomax distribution using non-informative 
priors, namely the Jeffreys prior \cite{jeffreys61} and the reference prior \cite{bernardo79}. Then, we perform a simulation study to compare the efficiency of the Bayesian approach for estimating the model parameters under these two priors, and check for the possible impact of prior specification. We also show how to represent the Lomax distribution in a hierarchical form by augmenting the model with a latent variable which makes the Bayesian computations easier to implement. This would also allow the user to implement inferences using all-purpose Bayesian statistical packages like WinBUGS \cite{lunn2000} or JAGS \cite{plummer2003}.

The remainder of this paper is organized as follows. In Section \ref{lomax}, we present the Lomax distribution and list some of its properties. In Section \ref{priori}, we formulate the Bayesian model using non-informative priors. In Section \ref{simulations}, a simulation study is  presented. In Section \ref{app}, the methodology is illustrated on a real data set. Some final comments are given in Section \ref{conclusions}.

\section{Model definition} \label{lomax}

Here, we use the definition that appears, for example, in \cite{Howlader2002}.

\begin{definition}
A continuous random variable $X$ has a Lomax distribution with parameters $\alpha$ and $\beta$ if its probability density
function is given by $$f(x|\beta,\alpha)=\frac{\alpha}{\beta}\left(1+\frac{x}{\beta}\right)^{-(\alpha+1)}, \quad x\ge0,$$
where $\alpha>0$ and $\beta>0$ are the shape and scale parameters, respectively.
\end{definition}

We refer to this distribution as $Lomax\left(\beta,\alpha\right)$. The median is $\beta(2^{1/\alpha}-1)$ and the mode is zero. The hazard function is given by 
\[
h(x|\beta,\alpha) = \frac{\alpha}{\beta}\left(1+\frac{x}{\beta}\right)^{-1}, \quad x\ge0,
\]
which is a decreasing function of $x$, thus making this a suitable model for components that age with time. The survival function is given by
\[
S(x|\beta,\alpha) = \left(1+\frac{x}{\beta}\right)^{-\alpha}, \quad x\ge0.
\]

We note also that the Lomax distribution can be expressed in the following hierarchical form
\begin{eqnarray*}
X|\beta,\lambda &\sim & Exponential\left(\frac{\lambda}{\beta}\right),\\
\lambda|\alpha  &\sim & Gamma(\alpha,1).
\end{eqnarray*}

This follows from writing the joint density of $X$ and $\lambda$ as
\[
f(x| \beta,\lambda)f(\lambda|\alpha)= \frac{1}{\beta\Gamma(\alpha)} 
\lambda^{\alpha}\exp\left\{-\lambda\left(1+\frac{x}{\beta}\right)\right\}.
			\]
So, the marginal density of $X$ is given by
\begin{eqnarray*}
f(x) &=&  \frac{1}{\beta\Gamma(\alpha)} \int_0^{\infty} \lambda^{\alpha}
\exp\left\{-\lambda\left(1+\frac{x}{\beta}\right)\right\}d\lambda \\
&=& \frac{1}{\beta\Gamma(\alpha)} \Gamma(\alpha+1)\left(1+\frac{x}{\beta}\right)^{-(\alpha+1)}\\
&=& \frac{\alpha}{\beta} \left(1+\frac{x}{\beta}\right)^{-(\alpha+1)},
\end{eqnarray*}
and we can conclude that $X\sim Lomax(\beta,\alpha)$.

Using this mixture representation, it is not difficult to see that the
unconditional mean and variance of $X$ are given by 
\begin{eqnarray*}
E(X) 
&=& 
\beta E[\lambda^{-1}]=\frac{\beta}{\alpha-1}, \quad \alpha>1,\\\\
Var(X) 
&=& 
\beta^2\left\{E\left[\lambda^{-1}\right]^2 + Var\left[\lambda^{-1}\right]\right\}
=
\frac{\alpha\beta^2}{(\alpha-1)^2(\alpha-2)}, \quad \alpha>2,
\end{eqnarray*}
since $\lambda^{-1}\sim IG(\alpha,1)$, where $IG(a,b)$ denotes the Inverse Gamma distribution with parameters
$a$ and $b$, mean $b/(a-1)$, $a>1$, and variance $b^2/(a-1)^2 (a-2)$, $a>2$.

Now, suppose that $\textbf{X}=(X_1,\dots,X_n)$ is a random sample of size $n$ from the Lomax distribution. 
We assume that the mixing parameters $\boldsymbol{\lambda}=(\lambda_1,\dots,\lambda_n)$ are a priori independent. The
complete conditional distribution of $\boldsymbol{\lambda}$ using the hierarchical form is given by
\begin{eqnarray*}
f(\boldsymbol{\lambda}|\textbf{x},\beta,\alpha)
&\propto& 
f(\textbf{x}|\beta,\boldsymbol{\lambda})~f(\boldsymbol{\lambda}|\alpha)\\
&\propto&
\prod_{i=1}^n\lambda_i\exp(-\lambda_ix_i/\beta)
\prod_{i=1}^n\lambda_i^{\alpha-1}\exp(-\lambda_i)\\
&\propto&
\prod_{i=1}^n\lambda_i^{\alpha}
\exp\left\{-\lambda_i\left(1+\frac{x_i}{\beta}\right)\right\},  
\end{eqnarray*}
so that
$$
\lambda_i|\textbf{x},\boldsymbol{\lambda}_{-i},\alpha,\beta\sim
Gamma\left(\alpha+1,1+\frac{x_i}{\beta}\right).
$$
Again, using the hierarchical form, we obtain the complete conditional distributions of $\alpha$ and $\beta$ as
\begin{eqnarray*}
f(\alpha|\textbf{x},\boldsymbol{\lambda},\beta) &\propto& 
f(\boldsymbol{\lambda}|\alpha)~\pi(\beta,\alpha) ~ \propto ~
\left[\Gamma(\alpha)\right]^{-n} \left(\prod_{i=1}^n\lambda_i\right)^{\alpha-1}\pi(\beta,\alpha),\\
f(\beta|\textbf{x},\boldsymbol{\lambda},\alpha)
&\propto& f(\textbf{x}|\beta,\lambda)~\pi(\beta,\alpha) ~ \propto ~
\beta^{-n}\exp\left\{-\frac{1}{\beta}\sum_{i=1}^n\lambda_ix_i\right\}\pi(\beta,\alpha).
\end{eqnarray*}

We note that, using this representation of the Lomax distribution, each observation $X_i$ is associated with one mixing parameter $\lambda_i$, whose posterior mean or median can be used to identify a possible outlier.

\section{Prior specification} \label{priori}

We now complete the Bayesian model by specifying a prior distribution for $\alpha$ and $\beta$. We consider non-informative priors on these parameters and verify the existence of their posterior distribution.

\subsection{Jeffreys prior}

A commonly used objective prior in Bayesian analysis is Jeffreys prior \cite{jeffreys61}, which is defined as
\[   \pi_{\tt J}(\beta,\alpha)\propto |I(\beta,\alpha)|^{1/2}, 
	\]
where $I(\cdot)$ stands for the Fisher information matrix. This is given by
\begin{eqnarray*}%\label{MIF}
I(\beta,\alpha) = n\left[\begin{array}{cc}
\dfrac{\alpha}{\beta^2(\alpha + 2)} & \quad -\dfrac{1}{\beta(\alpha + 1)} \\\\
-\dfrac{1}{\beta(\alpha + 1)} & \quad \dfrac{1}{\alpha^2}
\end{array}\right],
\end{eqnarray*}
from which we obtain
\begin{eqnarray*} %\label{prior_jeff}
\pi_{\tt J}(\beta, \alpha) \propto 
\frac{1}{\beta(\alpha+1) \alpha^{1/2}(\alpha+2)^{1/2}}, \quad \beta, \alpha>0.
\end{eqnarray*}

Considering independence between the parameters, the Jeffreys joint prior for $(\beta,\alpha)$ is given by 
$\pi_{\tt IJ}(\beta,\alpha)\propto\pi(\beta)\pi(\alpha)=1/\beta\alpha$. 

Then, substituting $\pi(\beta,\alpha)$ in the expressions for the complete conditional densities, we obtain
\begin{eqnarray*}
f(\alpha|\textbf{x},\boldsymbol{\lambda},\beta) \propto
\frac{1}{(\alpha+1)\alpha^{1/2}(\alpha+2)^{1/2}\Gamma^n(\alpha)}
\left(\prod_{i=1}^n\lambda_i\right)^{\alpha-1}
\end{eqnarray*}
for the dependent Jeffreys prior and
$$
f(\alpha|\textbf{x},\boldsymbol{\lambda},\beta) \propto
\frac{1}{\alpha\Gamma^n(\alpha)}\left(\prod_{i=1}^n\lambda_i\right)^{\alpha-1}
$$
for the independence case. So, the complete conditional distribution of $\alpha$ is not of standard form and a Metropolis-Hastings algorithm \cite{Hastings70,Chib95} is used to sample its values. The complete conditional density of $\beta$ is given by
$$
f(\beta|\textbf{x},\boldsymbol{\lambda},\alpha)
\propto
\beta^{-(n+1)}\exp\left\{-\frac{1}{\beta}\sum_{i=1}^n\lambda_ix_i\right\}
$$
for both dependent and independent Jeffreys priors. It then follows that
$$
\beta|\textbf{x},\boldsymbol{\lambda},\alpha\sim IG\left(n,\sum_{i=1}^n\lambda_ix_i\right).
$$

\subsection{Reference prior}

Reference priors were first proposed by \cite{bernardo79} and further developed by \cite{berger89,berger92a,berger92b,clarke1997,berger2009a,berger2009b}, among others. The idea is to specify a prior distribution such that, even for moderate sample sizes, the information provided by the data should dominate the prior information. In particular, \cite{berger92a} discuss the construction of a non-informative prior that gives a different treatment for parameters of interest and nuisance parameters. When there are nuisance parameters (which is the case in this paper), one must establish an order parametrization between the interest and nuisance parameters. The following proposition is borrowed from \cite{bernardo05} and adapted to the Lomax case.

\begin{proposition} \label{prep2a}
Let $f\left(\mathbf{x}|\beta,\alpha\right)$, $(\beta,\alpha)\in{\rm \Delta}\times\alpha(\beta)\subseteq$ ${\rm I\!R}\times{\rm I\!R}$
be a probability model. Suppose that the joint posterior distribution of $(\beta,\alpha)$\ is asymptotically normal with
covariance matrix $S(\hat{\beta},\hat{\alpha})=I^{-1}(\hat{\beta},\hat{\alpha})$, where $\hat{\beta}$ and $\hat{\alpha}$ are consistent estimators of $\beta$ and $\alpha$, respectively. Then, if $\beta$ is the parameter of interest and $\alpha$ is the nuisance parameter,
\begin{itemize}
\item [(i)] The conditional reference prior of $\alpha$ given $\beta$ is
\[
\pi\left(\alpha | \beta \right)  \propto 
\left[I_{22}(\beta,\alpha)\right]^{1/2}, \quad \alpha\in\alpha\left(\beta\right).
\]
\item [(ii)] If $\pi\left(\alpha|\beta\right)  $ is not proper, a compact approximation $\left\{\alpha_{i} \left(\beta\right), \mbox{ } i = 1, 2, \dots\right\}$ to $\alpha\left(\beta\right)$ is required and the reference prior of $\alpha$ given $\beta$ is
\[
\pi_{i}\left(\alpha|\beta\right)  =\frac{\left[I_{22}(\beta,\alpha)\right]^{1/2}}%
{\int\limits_{\alpha_{i}\left(\beta\right)}\left[I_{22}(\beta,\alpha)\right]^{1/2}d\alpha},\quad \alpha\in\alpha_{i}\left(
\beta\right).
\]
\item [(iii)] The sequence of priors can be obtained as
\[
\pi_{i}(\beta)\propto\exp\left\{\int\limits_{\alpha_{i}\left(  \beta\right)
}\pi_{i}\left(\alpha|\beta\right)  \log\left[s_{11}^{1/2}\left(
\beta,\alpha\right)  \right]  d\alpha\right\},
\]
where $s_{11}^{1/2}\left(\beta,\alpha\right)= I_{\beta}\left(\beta,\alpha\right)  \mathbf{=}I_{11}-I_{12}I_{22}^{-1}I_{21}$.
\item [(iv)] The reference posterior distribution of $\beta$ given data $\mathbf{x}=(x_1,\dots,x_n)$ is
\[
\pi(\beta|\mathbf{x})\propto
\pi(\beta)\left\{
\int\limits_{\alpha\left(\beta\right)}
\left[\prod\limits_{i=1}^{n}f(x_{i}|\beta,\alpha)\right]
\pi(\alpha|\beta) d\alpha\right\}.
\]
\end{itemize}
\end{proposition}

\begin{proof}
See a heuristic justification in \cite{bernardo05}.
\end{proof}

\begin{corollary}\label{coro1}
 If the nuisance parameter space $\alpha\left(\beta\right) = \alpha$ is independent of $\beta$, and the functions
$s_{11}^{-1/2}\left(\beta,\alpha\right)  $ and $I_{22}^{1/2}(\beta,\alpha)$
factorize in the form
\[
\left[s_{11}\left(\beta,\alpha\right)\right]^{-1/2} = f_{1}\left(
\beta\right)  g_{1}\left(\alpha\right), \qquad \left[I_{22}%
(\beta,\alpha)\right]^{1/2} = f_{2}\left(\beta\right)  g_{2}\left(\alpha\right),
	\]
then
\[
\pi(\beta)\propto f_{1}\left(\beta\right) \quad \mbox{and} \quad
\pi\left(\alpha|\beta\right)  \propto g_{2}\left(\alpha\right).
	\]
Thus, the reference prior relative to the ordered parametrization $(\beta,\alpha)$ is given by
\[
\pi_{\tt R}(\beta,\alpha)=\pi(\beta)\pi(\alpha|\beta)=f_{1}\left(\beta\right)g_{2}(\alpha).
	\]
\end{corollary}

\begin{proof}
See proof of Theorem 12 in \cite{bernardo05}.
\end{proof}

\begin{proposition} %\label{prop4.2}
The joint reference prior for the Lomax model with parameters $\beta$ and $\alpha$, 
$(\beta,\alpha)\in{\rm \Delta}\times\alpha(\beta)\subseteq$ ${\rm I\!R}\times{\rm I\!R}$, is given by
\begin{eqnarray*} \label{prior_ref}
\pi_{\tt R}(\beta,\alpha) \propto \frac{\pi(\beta)}{\alpha},
\end{eqnarray*}
where $\pi(\beta) \propto f_1(\beta)$.
\end{proposition}
\begin{proof}
The inverse of the Fisher information matrix is given by
\begin{eqnarray*}
I^{-1}(\beta, \alpha) = S(\beta, \alpha) = \frac{1}{n}\left[\begin{array}{cc}
\dfrac{\beta^2(\alpha + 2)(\alpha + 1)^2}{\alpha} & \quad \beta\alpha(\alpha + 2)(\alpha + 1)\\\\
\beta\alpha(\alpha + 2)(\alpha + 1) & \alpha^2(\alpha+1)^2
\end{array}\right],
\end{eqnarray*}
from which we obtain
\begin{eqnarray*}
\left[s_{11}(\beta,\alpha)\right]^{-1/2} \propto
\left[\frac{\beta^2(\alpha + 2)(\alpha + 1)^2}{\alpha}\right]^{-1/2} = 
\frac{\alpha^{1/2}}{\beta(\alpha + 2)^{1/2}(\alpha + 1)}.
\end{eqnarray*}
Considering Corollary \ref{coro1},
\begin{eqnarray*}
f_1(\beta) = \frac{1}{\beta} \quad \mbox{and} \quad g_1(\alpha) = \frac{\alpha^{1/2}}{(\alpha + 2)^{1/2}(\alpha + 1)}.
\end{eqnarray*}
From the Fisher information matrix, we also have that
\begin{eqnarray*}
\left[I_{22}(\beta, \alpha)\right]^{1/2} = \left[\frac{1}{\alpha^2}\right]^{1/2} = \frac{1}{\alpha},
\end{eqnarray*}
and finally, $f_2(\beta) = 1$ and $g_2(\alpha) = 1/\alpha$.
\end{proof} 

\noindent So, considering $\beta$ as the parameter of interest and $\alpha$ as the nuisance parameter, we conclude that the joint reference prior for $(\beta,\alpha)$ is given by
\[
\pi_{\tt R}(\beta,\alpha)\propto f_1(\beta)g_2(\alpha)=\frac{1}{\alpha\beta}, \quad \beta,\alpha>0.
	\] 
Note that this reference prior coincides with the Jeffreys prior under the assumption of independence of parameters ($\pi_{\tt IJ}(\beta,\alpha)$) in the
Lomax distribution. Consequently, the complete conditional distributions are the same as for the independent Jeffreys prior.

Applying the Bayes theorem, the joint posterior
density is given by
\begin{eqnarray}\label{post}
\pi(\beta,\alpha|\textbf{x})
&\propto &
\left(\frac{\alpha}{\beta}\right)^n 
\prod_{i=1}^{n} \left(1 + \frac{x_i}{\beta}\right)^{-(\alpha+1)}\times\frac{1}{\alpha\beta}\nonumber\\
&\propto &
\alpha^{n-1}  \beta^{-(n+1)}\prod_{i=1}^{n} \left(1 + \frac{x_i}{\beta}\right)^{-(\alpha+1)}.
\end{eqnarray}

\begin{proposition} \label{proper2}
The posterior distribution (\ref{post}) is improper for $n=1$  and proper for $n>1$.
\end{proposition}

\begin{proof}
See Appendix \ref{A}.
\end{proof}

\section{Simulation study} \label{simulations}

In this section, we perform a Monte Carlo study to evaluate the methodology described in the previous section. We generated $m = 500$ replications of samples of sizes $n = 50$, $100$, $150$, $200$, $300$ and $500$ from the Lomax distribution, considering parameter values $\beta = 2$ and $\alpha = 1.5$. The model was then estimated using Jeffreys and reference priors. We used the Metropolis-Hastings algorithm and the Gibbs sampler implemented in software R \cite{r2014} to simulate two chains of values from the posterior distribution. A total of $11,000$ iterations with jumps of 10 and a burn-in of $1,000$ were performed for each chain, thus leading to a final sample of $1,000$ values for each chain. Also, the Gelman and Rubin's Monte Carlo Markov Chain (MCMC) convergence diagnostic \cite{gr92} provided in the R package CODA \cite{coda2006} was used to monitor convergence of the two parallel chains. 

Let $\hat{\theta}^{(j)}$ be the estimate of parameter $\theta$ for the $j$-th replication, $j=1,\ldots,m$. These are the parameter posterior means calculated from the $2,000$ simulated values for each replication. To evaluate the estimation method, two criteria were considered: the bias and square root of the mean square error (or simply root mean square error, rmse), which are defined as

\begin{eqnarray*}
\mbox{bias}  &=&\frac{1}{m}\sum_{j=1}^{m}\hat{\theta}^{(j)} - \theta, \\
\mbox{rmse}  &=&\sqrt{\frac{1}{m}\sum_{j=1}^{m}\left ( \hat{\theta}^{(j)} - \theta\right)^2}.
\end{eqnarray*}

The results from the simulated experiment appear in Tables
\ref{tab1}-\ref{tab4} and in Figure \ref{fig1}. Tables \ref{tab1} and
\ref{tab3} show the computed posterior mean, standard deviation (SD),
95\% credibility interval (CI), bias and rmse for each parameter  and
sample size, considering Jeffreys and reference priors,
respectively. We note that the performances are barely similar
whatever the prior we adopt. For the parameter of interest $\beta$ the
bias is relatively small and negative for both priors, while for the
nuisance parameter $\alpha$ the bias becomes positive for sample sizes
$n\ge 100$ when a Jeffreys prior is adopted. With respect to accuracy,
we obtained good results for $\beta$ and $\alpha$ with relatively
small rmse for both Jeffreys and reference priors and moderate sample
sizes. Overall, the biases tend to reduce when moving from a Jeffreys
to a reference prior, however this is mainly observed in the rmse for
sample sizes $n\ge 200$ (see Figure \ref{fig1}). The Gelman and
Rubin's diagnostic values presented in Tables \ref{tab2} and
\ref{tab4} are close to 1, which means that the algorithm converged,
independent of the initial values adopted. Note also from these tables
that the acceptance rates were better (higher) when using the
reference prior. 

\begin{table}[h]
\centering
\caption{Results of the simulation study using Jeffreys prior,
  considering true values of parameters:  $\beta = 2$ and $\alpha =
  1.5$.}\label{tab1}
  \begin{tabular}{lcccccc}
 \hline
sample size     & parameter& mean & SD & 95\% CI  &bias & rmse\\ \hline
$50$ &$\beta$ & 3.5330&	2.2255&	[0.8639 ; 9.2080] &-1.5330&	5.6185\\
     &$\alpha$& 2.1413& 0.9949&	[0.8596 ; 4.6271] &-0.6413&	2.5356\\
$100$&$\beta$ & 2.2017&	0.9074&	[1.0549 ; 4.2964] &-0.7017&	1.1643\\
     &$\alpha$& 1.7865&	0.4875&	[1.0966 ; 2.9148] & 0.2135&  0.5394\\
$150$&$\beta$ & 2.1693&	0.6548&	[1.1611 ; 3.7067] &-0.6693&  0.9377\\
     &$\alpha$& 1.6383&	0.3316&	[1.0897 ; 2.4156] & 0.3616&	0.4912\\
$200$&$\beta$ & 2.2161&	0.6528&	[1.2327 ; 3.7689] &-0.7161&	1.2285\\
     &$\alpha$& 1.9017&	0.4812&	[1.1631 ; 3.0344] & 0.0983&	0.8494\\
$300$&$\beta$ & 2.1263&	0.4984&	[1.3307 ; 3.2698] &-0.6263&	0.9486\\
     &$\alpha$& 1.8401&	0.3697&	[1.2423 ; 2.6795] & 0.1599&	0.6497\\
$500$&$\beta$ & 2.0498&	0.3326&	[1.4456 ; 2.7474] &-0.5498&	0.6428\\
     &$\alpha$& 1.4213&	0.1508&	[1.1496 ; 1.7245] & 0.5787&        0.5980\\ 
\hline
\end{tabular}
\end{table}

\begin{table}[h]
  \centering
  \caption{Evaluation of the algorithm using Jeffreys prior.}	
  \label{tab2}
  \begin{tabular}{lcccccc}
 \hline
sample size  &$50$&	$100$&	$150$	&$200$ &$300$&$500$\\ \hline
acceptance rate &0.4553&  0.3375 &  0.2799  &0.2337  &0.1311 &0.1312\\
Gelman-Rubin    &1.0033&1.0068   &1.0045    &1.0025  &1.0003 &1.0005\\ 
\hline
\end{tabular}
\end{table}

\begin{table}[h]
  \centering
  \caption{Results of the simulation study using reference prior, considering true values of parameters: $\beta = 2$ and $\alpha = 1.5$.}	\label{tab3}
  \begin{tabular}{lcccccc}
 \hline
 sample size &parameter&mean & SD & 95\% CI &	bias&rmse\\ \hline
$50$ &$\beta$ & 3.4805& 1.9893&	[0.8876 ; 8.2304]&	-1.4805& 4.2210\\
       &$\alpha$& 2.0957&	0.8955&	[0.8705 ; 4.2278]&	-0.5957& 1.7641	\\
$100$&$\beta$ & 2.7262&	1.1961&	[1.0864 ; 5.6940]&	-0.7262& 2.5608	\\
       &$\alpha$& 1.8305&	0.5609&	[1.0096 ; 3.1789]&	-0.3306& 1.2756\\
$150$&$\beta$ & 2.3929&	0.8328&	[1.1860 ; 4.4392]&  -0.3929& 1.4407\\
       &$\alpha$& 1.6956&	0.4042&	[1.0791 ; 2.6592]&	-0.1956& 0.7083\\
$200$&$\beta$ & 2.2596&	0.6629&	[1.2539 ; 3.8479]&	-0.2596& 1.0196\\
       &$\alpha$& 1.6162&	0.3159&	[1.1099 ; 2.3456]&	-0.1161& 0.4858\\
$300$&$\beta$ & 2.2290&	0.5285&	[1.3876 ; 3.4468]&	-0.2290& 0.8646\\
       &$\alpha$& 1.6035&	0.2534&	[1.1825 ; 2.1718]&	-0.1035& 0.4137\\
$500$&$\beta$ & 2.0821&	0.3685&	[1.4618 ; 2.8919]&	-0.0821& 0.5267	\\
       &$\alpha$& 1.5378&	0.1795&	[1.2259 ; 1.9224]& -0.0378& 0.2593\\
\hline
\end{tabular}
\end{table}

\begin{table}[h]
 \centering
\caption{Evaluation of algorithm using reference prior.}  \label{tab4}
 \begin{tabular}{lcccccc}
 \hline
sample size  &$50$&	$100$&$150$	&$200$ &$300$&$500$\\ \hline
acceptance rate & 0.9317& 0.9109 & 0.8938 & 0.8785 & 0.8578& 0.8151\\ 
Gelman-Rubin    &1.0021 & 1.0003 &  1.0002& 1.0477 & 1.0008& 1.0019\\ \hline
	\end{tabular}
  \end{table}

\begin{figure}[h]
\centering
\subfloat{\includegraphics[width=0.75\textwidth]{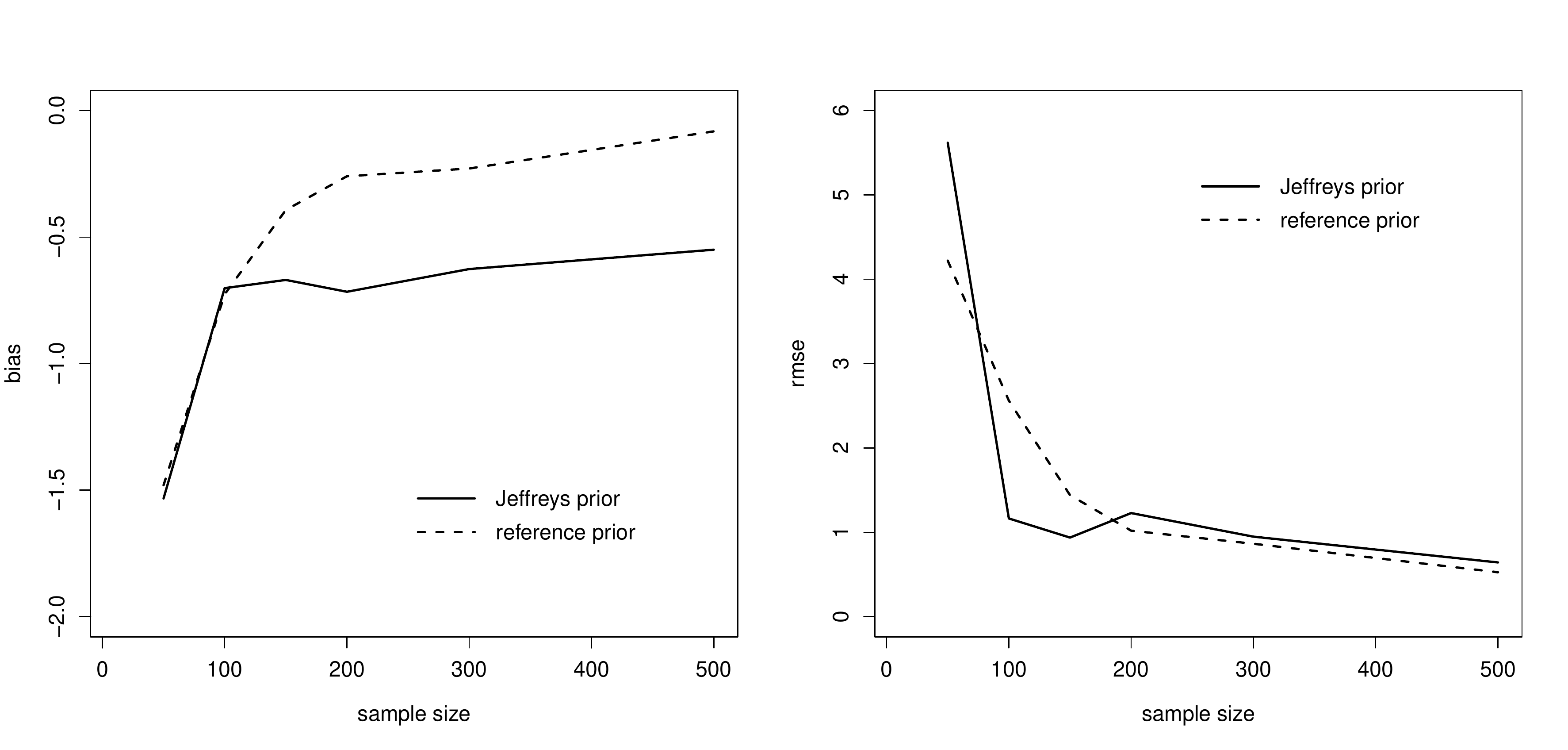}}
\qquad
\subfloat{\includegraphics[width=0.75\textwidth]{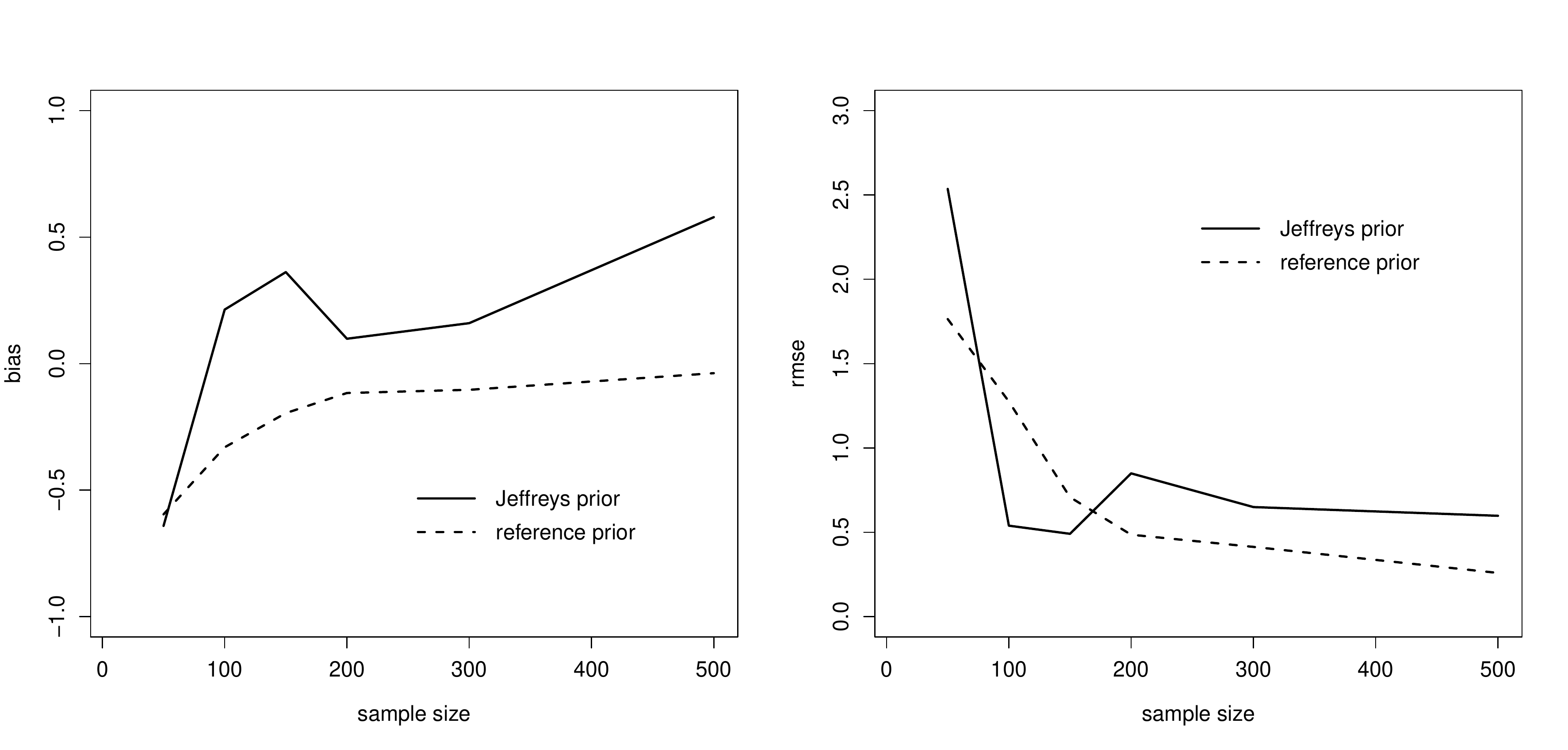}}
\caption{Bias and rmse for $\beta$ (upper panels) and $\alpha$ (lower panels) parameters using Jeffreys and reference priors.}
\label{fig1} 
\end{figure}

\section{Application} \label{app}

In order to illustrate the methodology proposed in this paper, we consider a sample of computer file sizes (in bytes) for all 269 files with the \textit{*.ini} extension on a Windows-based personal computer. These data can be downloaded from the website
{\tt http://web.uvic.ca/~dgiles/downloads/data}. A previous work by \cite{holland2006} has demonstrated the superiority of the Lomax distribution over several other competitors for modeling such file sizes. Those authors also provide technical information suggesting that the distribution should have infinite variance (i.e. $\alpha < 2$) in this context.

The samples for the Jeffreys and reference posterior distributions of the parameters $\beta$ and $\alpha$ were obtained by the Gibbs sampler and Metropolis-Hastings algorithm, i.e. through MCMC methods implemented in software R; see Appendix \ref{B}. The convergence of the chains were tested by using the Gelman and Rubin method implemented in the R package CODA. Graphical traces of those methods and kernel density estimation for each parameter showed that there were no convergence problems. We generated two parallel chains of size $80,000$ for each parameter. The first $20,000$ iterations were ignored to eliminate the effect of the initial values (burn-in) and, to avoid correlation problems, we considered a spacing of size 20, obtaining a sample of size $3,000$.

The posterior results from using both the Jeffreys and reference priors are shown in Table \ref{tab-real1}. It may be noticed that the posterior results are all very similar.

\begin{table}
  \centering
  \caption{Posterior summaries for the Lomax parameters using Jeffreys and reference priors.}
  \label{tab-real1}
  \begin{tabular}{ccccc}\hline  
prior &parameter& mean & SD & 95\% CI\\ \hline
\multirow{2}*{$\pi_{\tt J}$} &$\beta$  & 131.1242 &	24.5318 &	[88.9900 ; 184.6600] \\
&$\alpha$ & 0.5008 & 0.0435 &	[0.4207 ; 0.5920] \\ \hline
\multirow{2}*{$\pi_{\tt R}$} &$\beta$  & 130.4562 &	23.9599 &	[90.0000 ; 182.7000] \\
&$\alpha$ & 0.4986 & 0.0424 &	[0.4226 ; 0.5865] \\
\hline
 \end{tabular}
 \end{table}
	
In Figures \ref{fig1-res} and \ref{fig2-res} we show plots of the generated samples and the empirical marginal posteriors for model parameters $\beta$ and $\alpha$, based on the generated chains of the marginal Jeffreys and reference posteriors, respectively.

\begin{figure}[h]
\centering
\subfloat{\includegraphics[width=0.4\textwidth]{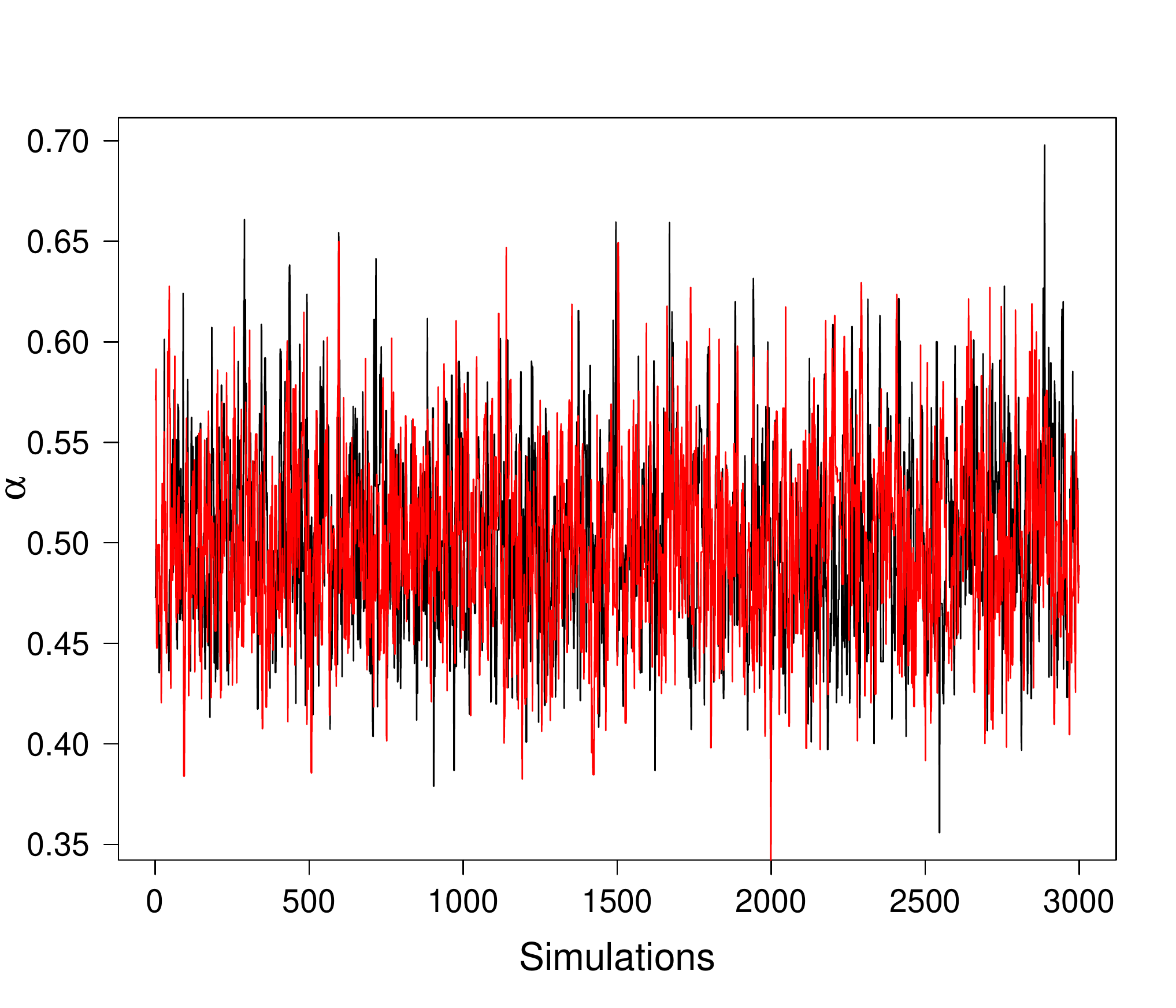}} \quad
\subfloat{\includegraphics[width=0.4\textwidth]{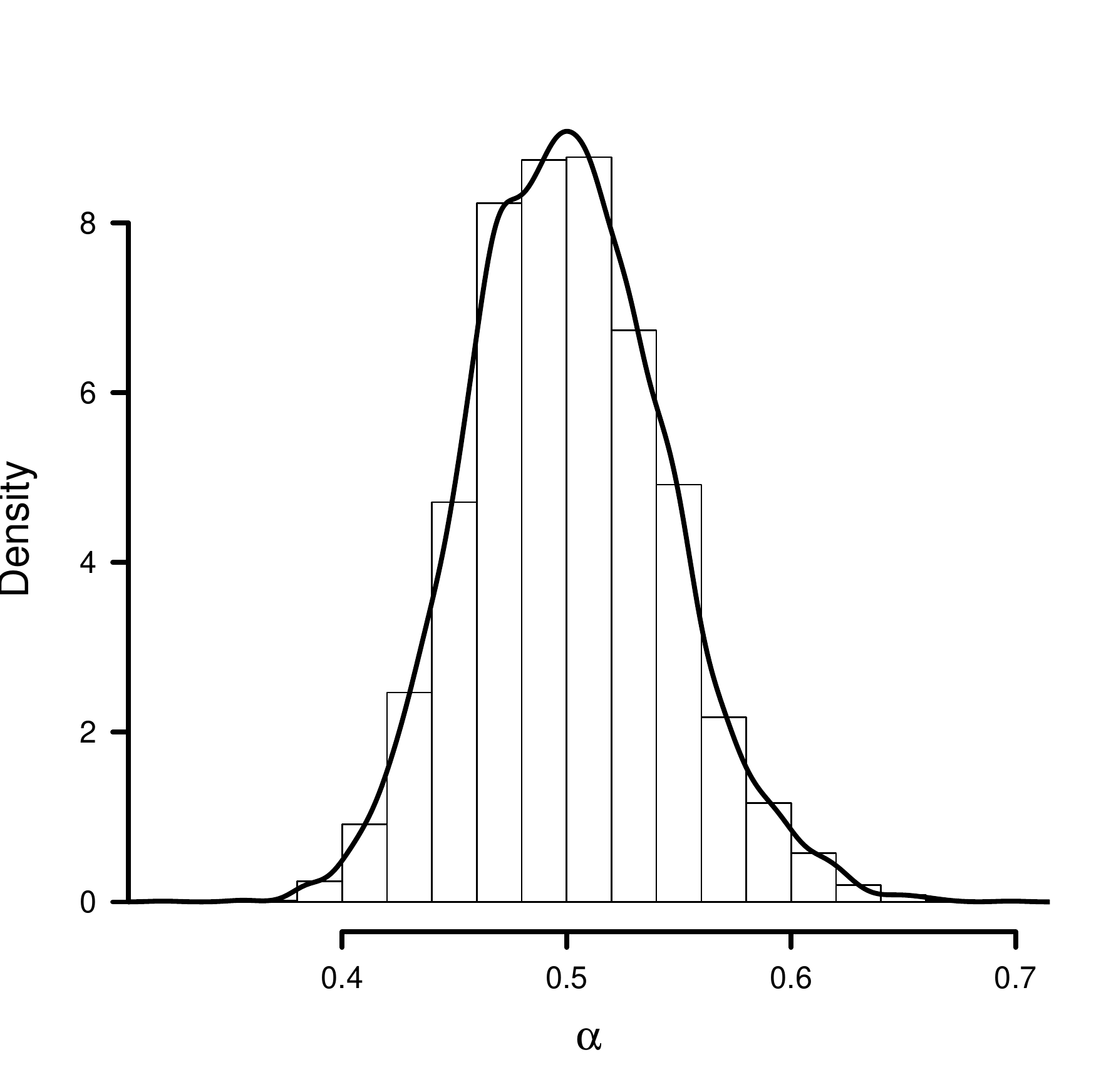}}
\qquad
\subfloat{\includegraphics[width=0.4\textwidth]{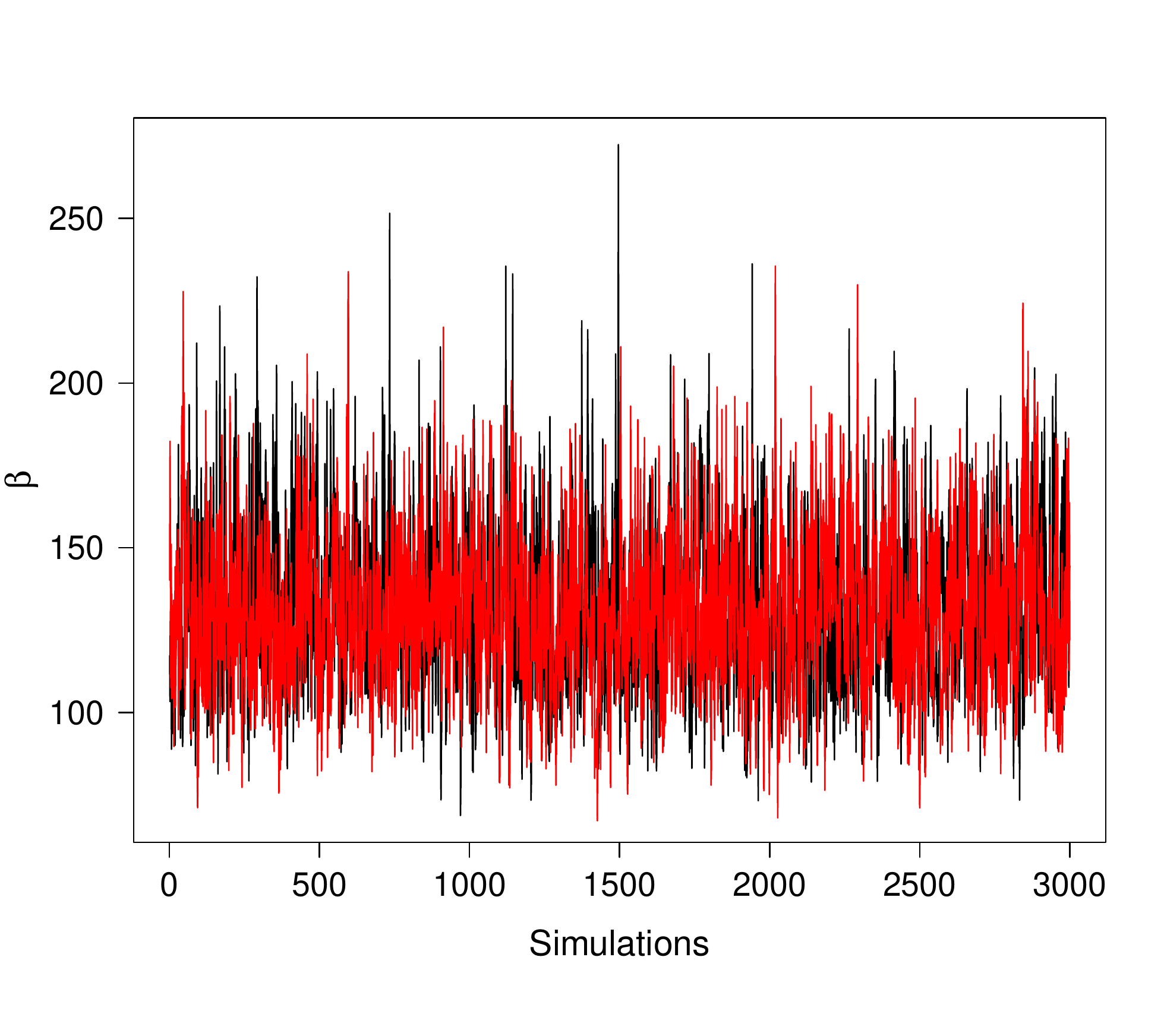}} \quad
\subfloat{\includegraphics[width=0.4\textwidth]{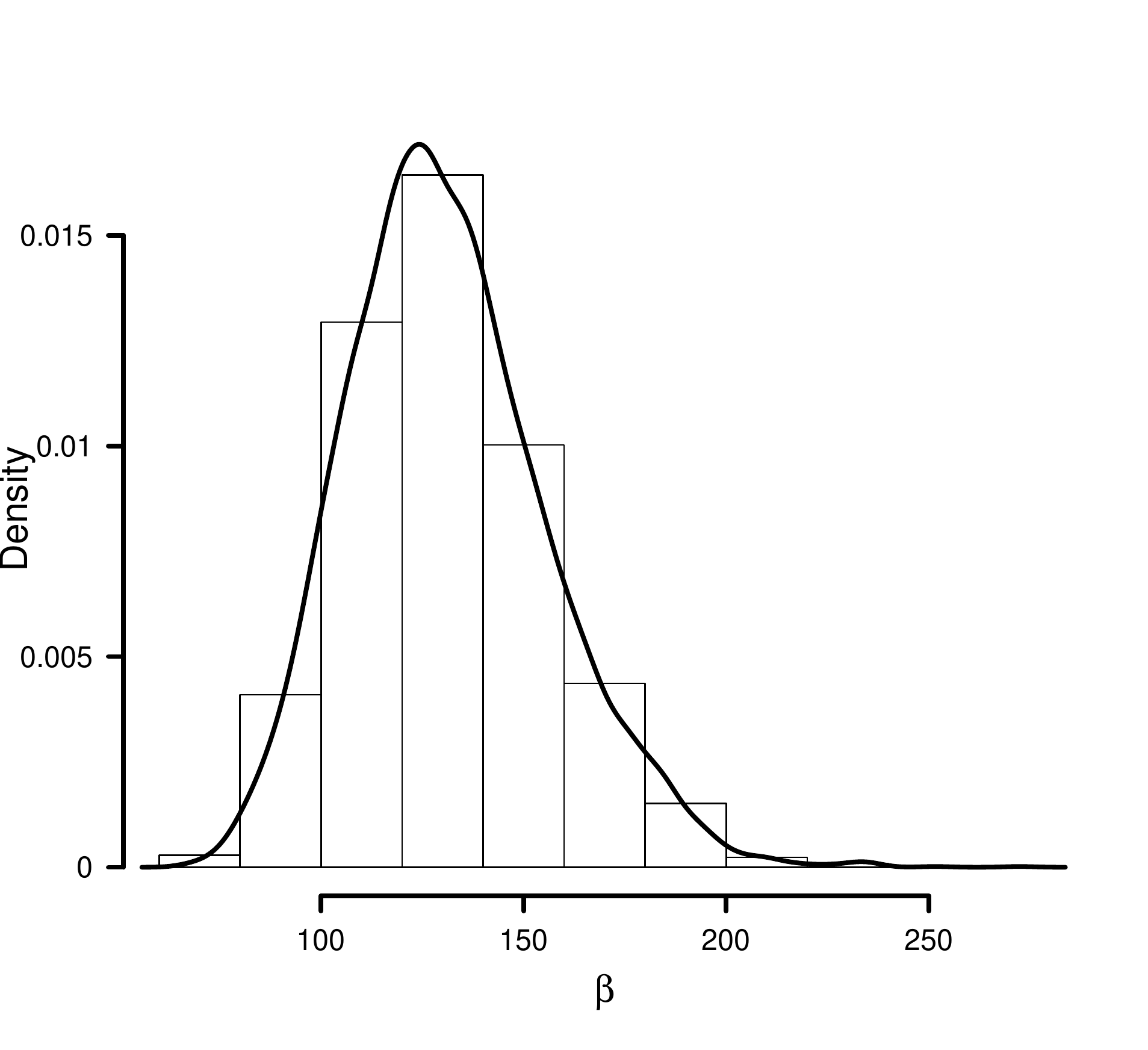}}
\caption{Trace and density of $\alpha$ (upper panels) and $\beta$ (lower panels) using Jeffreys prior.}
\label{fig1-res}
\end{figure}

\begin{figure}[h]
\centering
\subfloat{\includegraphics[width=0.4\textwidth]{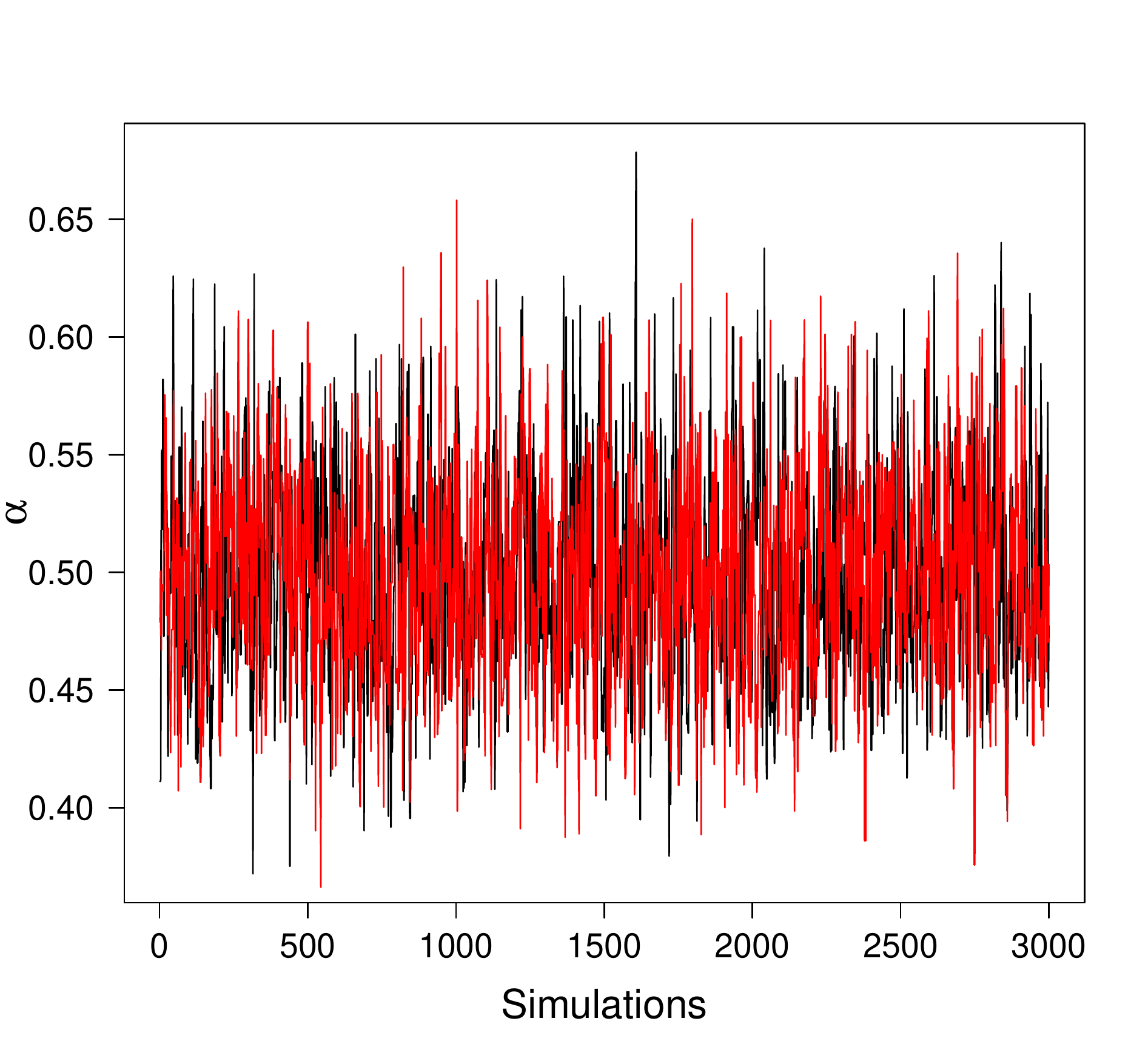}} \quad
\subfloat{\includegraphics[width=0.4\textwidth]{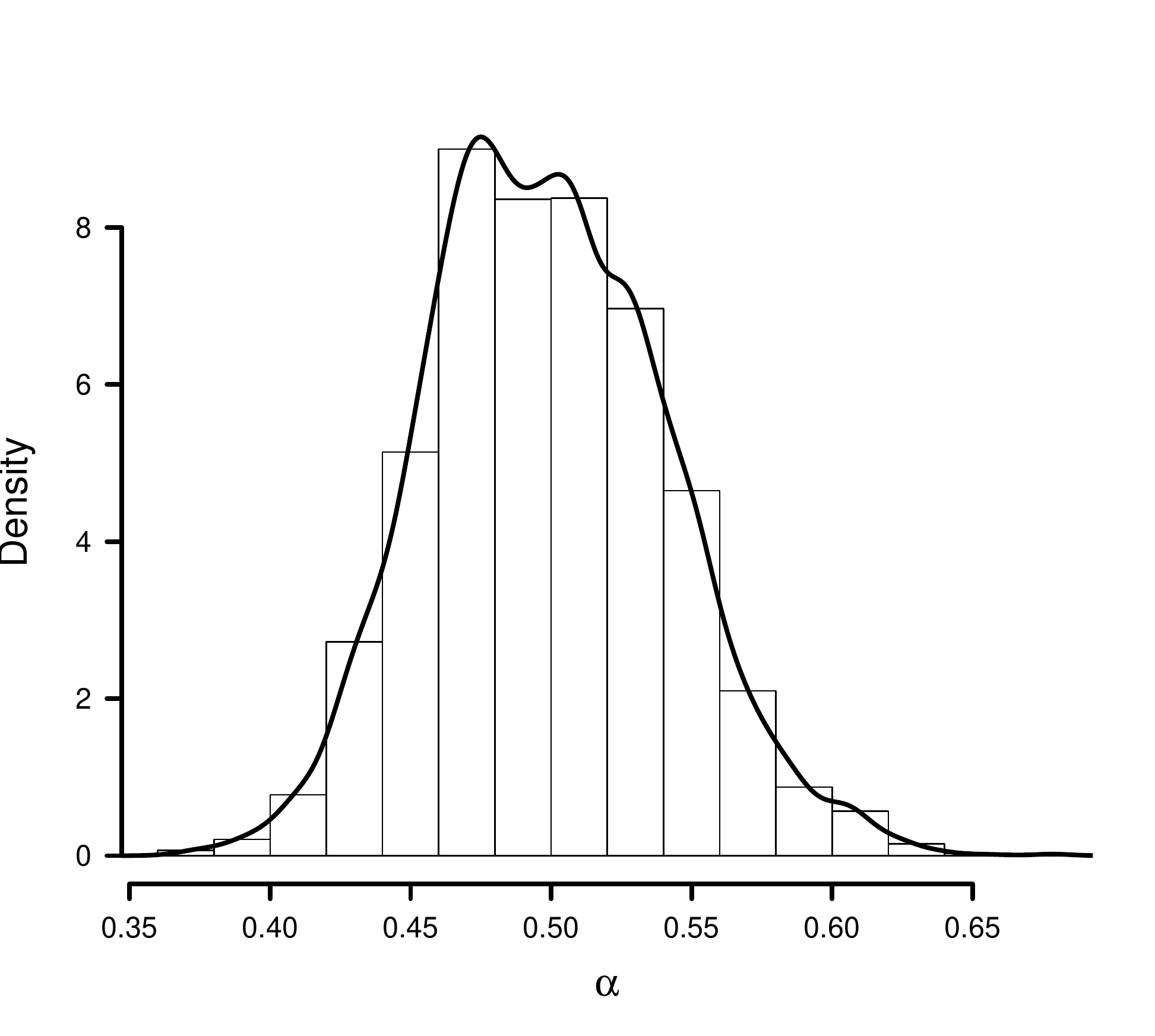}}
\qquad
\subfloat{\includegraphics[width=0.4\textwidth]{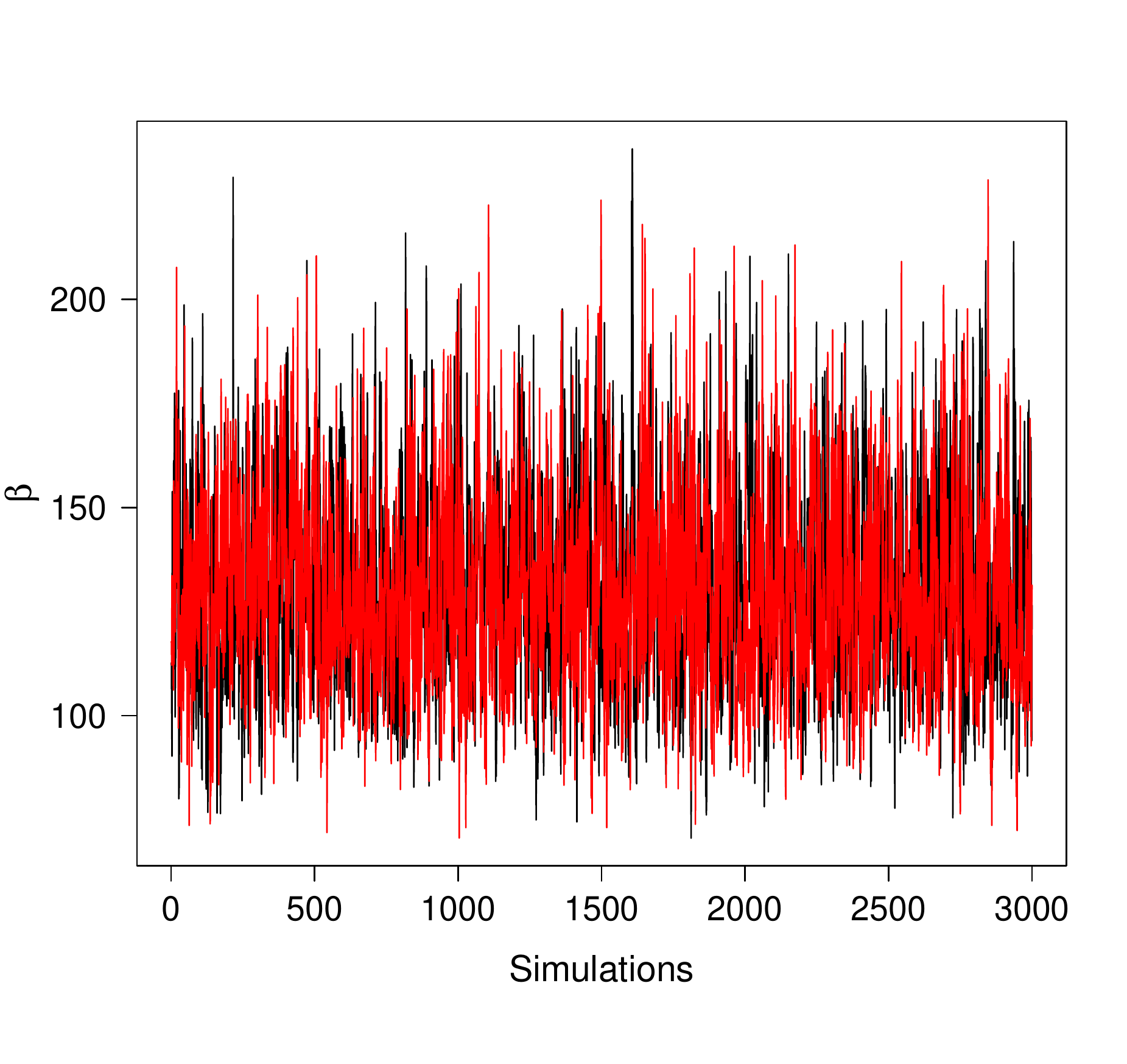}} \quad
\subfloat{\includegraphics[width=0.4\textwidth]{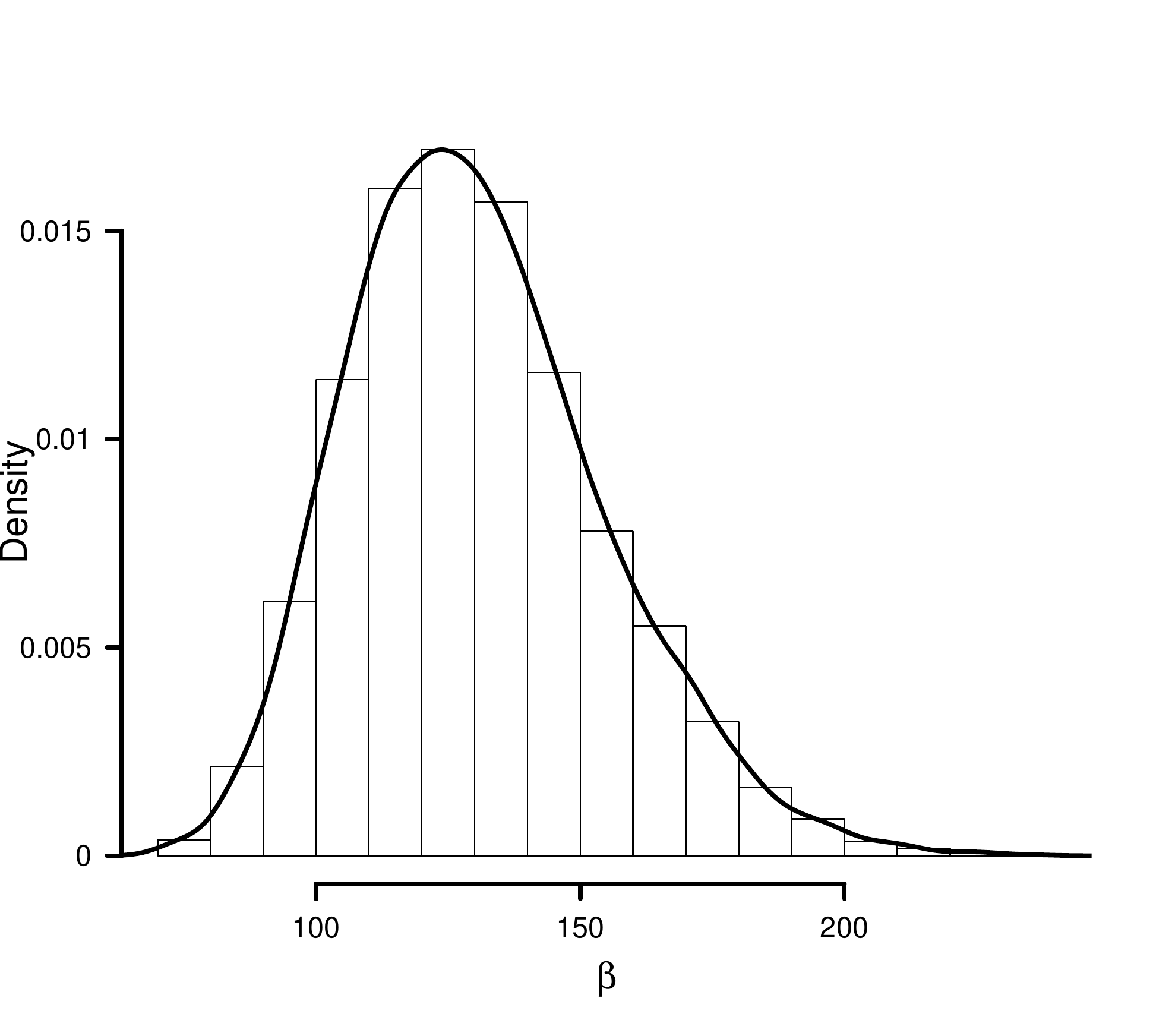}}
\caption{Trace and density of $\alpha$ (upper panels) and $\beta$ (lower panels) using reference prior.}
\label{fig2-res} 
\end{figure}

\section{Concluding remarks} \label{conclusions}

In this paper we evaluated the Bayesian method to estimate the parameters in a Lomax distribution under two non-informative prior
specifications. We showed that the joint posterior distribution is proper no matter which non-informative prior is used. 
We also obtained a scale mixture representation of the Lomax distribution in which the complete conditional distribution of the scale parameter is of known closed form and easy to sample. As a by product, this representation allows for the mixing parameters to
be used to identify possible outliers. 
Overall, the results obtained indicate that the Bayesian method estimates the parameters well under both prior specifications if the
sample size is not too small. Of course, as in any Monte Carlo study, our results are limited to our particular selection
of sample sizes and prior distributions. We hope that our findings are useful to the practitioners.

\section*{Acknowledgements}
The research of Francisco Louzada is funded by the Brazilian organization CNPq.

\clearpage
\appendix

\section{Appendix A}\label{A}

\subsection{Verifying that the posterior is proper under Jeffreys prior}

Under Jeffreys prior, the joint posterior density of $\beta$ and $\alpha$ is given by
$$
\pi(\beta,\alpha|\textbf{x})\propto
\frac{\alpha^{n-1/2} \beta^{-(n+1)}}{(\alpha+1)(\alpha+2)^{1/2}}
\prod_{i=1}^{n} \left(1 + \frac{x_i}{\beta}\right)^{-(\alpha+1)}.
$$
We next show that the integral of this expression is finite for any sample size $n$.

\begin{itemize}
\item [(a)] Verifying for $n=1$:
\begin{eqnarray*}
&&
\int_{0}^{\infty} \int_{0}^{\infty}
\frac{\alpha^{1/2} \beta^{-2}}{(\alpha+1)(\alpha+2)^{1/2}} \left(1 +
\frac{x}{\beta}\right)^{-(\alpha+1)} d\beta d\alpha \nonumber\\
&=& 
\int_{0}^{\infty} \frac{\alpha^{1/2}}{(\alpha+1)(\alpha+2)^{1/2}}
\left(\int_{0}^{\infty}\beta^{-2}\left(1 +
\frac{x}{\beta}\right)^{-(\alpha+1)}d\beta\right)
d\alpha \nonumber\\ 
&=& \int_{0}^{\infty} \frac{\alpha^{1/2}}{(\alpha+1)(\alpha+2)^{1/2}}
\left(\int_{0}^{\infty}\beta^{-2}\left(
\frac{\beta+x}{\beta}\right)^{-(\alpha+1)}d\beta\right)
d\alpha \nonumber\\
&=& \int_{0}^{\infty} \frac{\alpha^{1/2}}{(\alpha+1)(\alpha+2)^{1/2}}
  \left(\int_{0}^{\infty}\beta^{-2}\beta^{\alpha+1}\left(\beta+x\right)^{-(\alpha+1)}d\beta\right)
  d\alpha\\ 
&=& \int_{0}^{\infty} \frac{\alpha^{1/2}}{(\alpha+1)(\alpha+2)^{1/2}}
  \left(\int_{0}^{\infty}\beta^{\alpha-1}\left(\beta+x\right)^{-(\alpha+1)}d\beta\right)
  d\alpha \nonumber\\ 
&=& \int_{0}^{\infty} \frac{\alpha^{1/2}}{(\alpha+1)(\alpha+2)^{1/2}} \frac{1}{x\alpha} d\alpha \nonumber\\
&=& \frac{1}{x}\int_{0}^{\infty} \frac{\alpha^{-1/2}}{(\alpha+1)(\alpha+2)^{1/2}} d\alpha ~ = ~\frac{\pi}{2x} ~ < ~ \infty.
\end{eqnarray*}
		
\item [(b)] Verifying for $n>1$: 
First, we solve
\begin{eqnarray}\label{beta}
\int_{0}^{\infty} \beta^{-(n+1)}\prod_{i=2}^{n} \left(1 + \frac{x_i}{\beta}\right)^{-(\alpha+1)} d\beta.
\end{eqnarray}
Consider $y = \min(x_2, \dots, x_n)$. Then, it follows that
\begin{eqnarray*}	
\left(1 + \frac{x_i}{\beta}\right)^{\alpha+1}&\geq& 
\left(1 + \frac{y}{\beta}\right)^{\alpha+1}, \quad 
\alpha>0, \quad i=2,\dots,n\\\\ 
\prod_{i=2}^{n}\left(1 + \frac{x_i}{\beta}\right)^{\alpha+1}&\geq&
\left(1 + \frac{y}{\beta}\right)^{(n-1)(\alpha+1)}\\\\ 
\prod_{i=2}^{n}\left(1 + \frac{x_i}{\beta}\right)^{-(\alpha+1)}&<&
\left(1 + \frac{y}{\beta}\right)^{-(n-1)(\alpha+1)}. 
\end{eqnarray*}
Therefore,
\begin{eqnarray*}
&& \int_{0}^{\infty} \beta^{-(n+1)}\prod_{i=2}^{n} \left(1 +
\frac{x_i}{\beta}\right)^{-(\alpha+1)} d\beta ~ < ~
\int_{0}^{\infty} \beta^{-(n+1)} \left(1 +
\frac{y}{\beta}\right)^{-(n-1)(\alpha+1)} d\beta \\\\ 
&=&
\int_{0}^{\infty} \beta^{-(n+1)} \left(\frac{\beta+y}{\beta}\right)^{-(n-1)(\alpha+1)} d\beta \\\\
&=& 
\int_{0}^{\infty} 
\beta^{-(n+1)}\beta^{n\alpha+n-\alpha-1} \left(\beta+y\right)^{-(n-1)(\alpha+1)} d\beta\\\\
&=& 
\int_{0}^{\infty} 
\beta^{n\alpha-\alpha-2}\left(\beta+y\right)^{-(n-1)(\alpha+1)} d\beta 
~ = ~ \frac{(n-1)!\Gamma(n\alpha-\alpha-1)}{y^n\Gamma(n\alpha-\alpha+n-1)}.
\end{eqnarray*}

Then,
\begin{eqnarray} \label{pr1}
&&
\int_{0}^{\infty} \int_{0}^{\infty} 
\frac{\alpha^{n-1/2} \beta^{-(n+1)}}{(\alpha+1)(\alpha+2)^{1/2}}
\prod_{i=2}^{n} \left(1 + \frac{x_i}{\beta}\right)^{-(\alpha+1)}
d\beta d\alpha\nonumber\\
&=&
\int_{0}^{\infty}
\frac{\alpha^{n-1/2}}{(\alpha+1)(\alpha+2)^{1/2}}
\frac{(n-1)!\Gamma(n\alpha-\alpha-1)}{y^n\Gamma(n\alpha-\alpha+n-1)}d\alpha\nonumber\\
&=& 
\frac{(n-1)!}{y^n}\int_{0}^{\infty}  
\frac{\alpha^{n-1/2}}{(\alpha+1)(\alpha+2)^{1/2}}
\frac{\Gamma(n\alpha-\alpha-1)}{\Gamma(n\alpha-\alpha+n-1)}d\alpha\nonumber\nonumber\nonumber\\
&=& 
\frac{(n-1)!}{y^n}\int_{0}^{\infty} 
\frac{\alpha^{n-1/2}}{(\alpha+1)(\alpha+2)^{1/2}}\frac{1}{\prod_{j=-1}^{n-2}(n\alpha-\alpha+j)}d\alpha.
\end{eqnarray}

But note that
$(n\alpha-\alpha+j)\geq(n\alpha-\alpha)$, $\alpha>0$, $j\geq-1$.
Therefore,
\begin{eqnarray*}
(n\alpha-\alpha+j)^{-1} &<& (n\alpha-\alpha)^{-1}\\\\
\prod_{j=-1}^{n-2}(n\alpha-\alpha+j)^{-1} &<& \prod_{j=-1}^{n-2}(n\alpha-\alpha)^{-1}\\\\
\prod_{j=-1}^{n-2}(n\alpha-\alpha+j)^{-1} &<& (n\alpha-\alpha)^{-n+2}
\end{eqnarray*}

and replacing in (\ref{pr1}), we have that
$$
\frac{(n-1)!}{y^n}\int_{0}^{\infty} 
\frac{\alpha^{n-1/2}}{(\alpha+1)(\alpha+2)^{1/2}}\frac{1}{\prod_{j=-1}^{n-2}(n\alpha-\alpha+j)}d\alpha
$$
$$
<\frac{(n-1)!}{y^n}\int_{0}^{\infty} 
\frac{\alpha^{n-1/2}}{(\alpha+1)(\alpha+2)^{1/2}}(n\alpha-\alpha)^{-n+2}d\alpha ~ = ~ \infty.
$$
\end{itemize}
We can conclude that the posterior distribution using Jeffreys prior is proper for $n\ge 1$.

\subsection{Verifying that the posterior is proper under reference prior}

Using a reference prior, the joint posterior density of $\beta$ and $\alpha$ is given by
$$
\pi(\beta,\alpha|\textbf{x})\propto
\alpha^{n-1}  \beta^{-(n+1)}\prod_{i=1}^{n} \left(1 + \frac{x_i}{\beta}\right)^{-(\alpha+1)}.
$$
We next verify whether the integral of this expression is finite.

\begin{itemize}
\item [a)] Verifying for $n = 1$:
\begin{eqnarray*}
\int_{0}^{\infty} \int_{0}^{\infty} \beta^{-2}\left(1 +
\frac{x}{\beta}\right)^{-(\alpha+1)} d\beta d\alpha &=& \int_{0}^{\infty} \int_{0}^{\infty}
\beta^{-2}\left(\frac{\beta+x}{\beta}\right)^{-(\alpha+1)}d\beta d\alpha\\\\
&=& \int_{0}^{\infty} \int_{0}^{\infty}
\beta^{-2}\beta^{\alpha+1}\left(\beta+x\right)^{-(\alpha+1)}d\beta d\alpha\\\\	
&=& \int_{0}^{\infty} \int_{0}^{\infty}
\beta^{\alpha-1}\left(\beta+x\right)^{-(\alpha+1)}d\beta d\alpha\\\\
&=& \int_{0}^{\infty} \frac{1}{x\alpha} d\alpha ~ = ~ \infty.
\end{eqnarray*}
Therefore, the posterior distribution using a reference prior is improper for $n=1$.
	
\item [b)] Verifying for $n > 1$: Note that
$$
\int_{0}^{\infty} \beta^{-(n+1)}\prod_{i=2}^{n} \left(1 + \frac{x_i}{\beta}\right)^{-(\alpha+1)} d\beta
$$
is the same integral (\ref{beta}) previously resolved. Then, 
\begin{eqnarray*}
&&\int_{0}^{\infty} \int_{0}^{\infty}
\alpha^{n-1} \beta^{-(n+1)}\prod_{i=2}^{n} \left(1 +
\frac{x_i}{\beta}\right)^{-(\alpha+1)} d\beta d\alpha\\
&=&
\frac{(n-1)!}{y^n}\int_{0}^{\infty} \alpha^{n-1}\frac{1}{\prod_{j=-1}^{n-2}(n\alpha-\alpha+j)}d\alpha\\
&<&
\frac{(n-1)!}{y^n}\int_{0}^{\infty} \alpha^{n-1}(n\alpha-\alpha)^{-n+2}d\alpha ~ = ~ \infty.
\end{eqnarray*}
Thus, the posterior distribution using reference prior is proper for $n>1$.

\end{itemize}

\clearpage
%\appendixB
\section{{\tt R} code of the Metropolis-Hasting within Gibbs} \label{B}

\begin{lstlisting}
#==============================================================#
### clean memory
rm(list=ls(all=T))
graphics.off()
#==============================================================#
dat=read.table("dataset1.txt",header=T,sep=",",quote="\"",dec=".",
    fill=T,na.strings="NA")
n_dat=nrow(dat)
#==============================================================#
### Jeffreys (independence) and reference priors
## $p(\alpha,\beta) \propto \frac{1}{\alpha \beta}$
# use of the log of the conditional distribution of alpha
alpha_cond=function(x,lambda,tuning){
	repeat{
		# xn: proposal value
		xn=x+rnorm(1,0,tuning)
		if(xn>0)
		break
	}
	A1=length(lambda)*(lgamma(x)-lgamma(xn))+(xn-x)*
	   sum(log(lambda))+log(x)-log(xn)
	A2=pnorm(x,log.p=T)-pnorm(xn,log.p=T)
	A=A1+A2
	if(log(runif(1))<=A) x=xn else x=x
	return(x)
}
#==============================================================#
# some variable declaration
burn=20000
jump=20
n.amostra=3000
iter=burn+jump*n.amostra
#
ref_full=list(alpha1=numeric(),alpha2=numeric(),beta1=numeric(),
    beta2=numeric())
#
#==============================================================#
# MCMC method: Gibbs sampler with Metropolis-Hastings
# initial values
alpha_iter=rgamma(1,1,1)
beta_iter=rgamma(1,1,1)
#
tuning=3/3
#
for(j1 in 1:iter){
	# auxiliary variable
	lambda_aux=unlist(lapply(1:n_dat,function(i) rgamma(1,
	    shape=(alpha_iter+1),rate=((dat[i,1]/beta_iter)+1))))
	# conditional distribution of beta
	beta_iter=1/rgamma(1,shape=n_dat,scale=1/sum(lambda_aux*dat))
	# generation of values of alpha from its conditional
	# distribution
	alphan=alpha_cond(alpha_iter,lambda_aux,tuning)
	#
	alpha_iter=alphan
	#
	ref_full$alpha1[j1]=alpha_iter
	ref_full$beta1[j1]=beta_iter
}
#==============================================================#
\end{lstlisting}

\bibliographystyle{chicago}
%\bibliography{geral}

\clearpage

\begin{thebibliography}{10}
\providecommand{\url}[1]{{#1}}
\providecommand{\urlprefix}{URL }
\expandafter\ifx\csname urlstyle\endcsname\relax
  \providecommand{\doi}[1]{DOI~\discretionary{}{}{}#1}\else
  \providecommand{\doi}{DOI~\discretionary{}{}{}\begingroup
  \urlstyle{rm}\Url}\fi

\bibitem{bain92}
Bain, L., Engelhardt, M.: Introduction to Probability and Mathematical
  Statistics.
\newblock Duxbury Press (1992)

\bibitem{berger89}
Berger, J.O., Bernardo, J.M.: {Estimating a Product of Means: Bayesian Analysis
  with Reference Priors}.
\newblock Journal of the American Statistical Association \textbf{84}(405),
  200--207 (1989)

\bibitem{berger92a}
Berger, J.O., Bernardo, J.M.: {On the development of reference priors}, pp.
  35--60.
\newblock Oxford University Press (1992)

\bibitem{berger92b}
Berger, J.O., Bernardo, J.M.: {Ordered group reference priors with applications
  to a multinomial problem}.
\newblock Biometrika \textbf{79}, 25--37 (1992)

\bibitem{berger2009b}
Berger, J.O., Bernardo, J.M., Sun, D.: The formal definition of reference
  priors.
\newblock The Annals of Statistics \textbf{37}(2), 905--938 (2009)

\bibitem{berger2009a}
Berger, J.O., Bernardo, J.M., Sun, D.: Natural induction: An objective bayesian
  approach.
\newblock RACSAM - Revista de la Real Academia de Ciencias Exactas, Fisicas y
  Naturales, Serie A, Matematicas \textbf{103}(1), 125--135 (2009)

\bibitem{bernardo79}
Bernardo, J.: {Reference posterior distributions for Bayesian inference (with
  discussion)}.
\newblock Journal of the Royal Statistical Society, Series B \textbf{41},
  113--147 (1979)

\bibitem{bernardo05}
Bernardo, J.M.: {Reference Analysis}, vol.~25, pp. 17--90.
\newblock Elsevier (2005)

\bibitem{bryson74}
Bryson, M.: {Heavy-tailed distributions: Properties and tests.}
\newblock Technometrics \textbf{16}, 61--68 (1974)

\bibitem{Chib95}
Chib, S., Greenberg, E.: {Understanding the Metropolis-Hasting Algorithm}.
\newblock American Statistician \textbf{49}(4), 327--335 (1995)

\bibitem{clarke1997}
Clarke, B., Sun, D.: Reference priors under the chi-squared distance.
\newblock Sankhya: The Indian Journal of Statistics, Series A (1961-2002)
  \textbf{59}(2), 215--231 (1997)

\bibitem{gr92}
Gelman, A., Rubin, D.B.: Inference from iterative simulation using multiple
  sequences.
\newblock Statistical Science \textbf{7}(4), 457--472 (1992)

\bibitem{Hastings70}
Hastings, W.: {Monte Carlo Sampling Methods Using Markov Chains and Their
  Applications}.
\newblock Biometrika \textbf{57}, 97--109 (1970)

\bibitem{holland2006}
Holland, O., Golaup, A., Aghvami: {Traffic characteristics of aggregated module
  downloads for mobile terminal reconfiguration}.
\newblock Biometrika \textbf{135}, 683--690 (2006)

\bibitem{Howlader2002}
Howlader, H., Hossain, A.M.: {Bayesian survival estimation of Pareto
  distribution of the second kind based on failure-censored data.}
\newblock Computational Statistics \& Data Analysis \textbf{38}, 301--314
  (2002)

\bibitem{jeffreys61}
Jeffreys, H.: {Theory of Probability}.
\newblock Oxford Univ. Press, Oxford (1961)

\bibitem{lomax54}
Lomax, K.: {Business Failures; Another example of the analysis of failure
  data}.
\newblock Journal of the American Statistical Association - JSTOR \textbf{49},
  847--852 (1954)

\bibitem{lunn2000}
Lunn, D.J., Thomas, A., Best, N., Spiegelhalter, D.: Win\textsc{BUGS} - a
  \textsc{B}ayesian modelling framework: Concepts, structure, and
  extensibility.
\newblock Statistics and Computing \textbf{10}(4), 325--337 (2000)

\bibitem{plummer2003}
Plummer, M.: {JAGS}: A program for analysis of {Bayesian} graphical models
  using {Gibbs} sampling.
\newblock In: Proceedings of the 3rd International Workshop on Distributed
  Statistical Computing (2003)

\bibitem{coda2006}
Plummer, M., Best, N., Cowles, K., Vines, K.: {CODA: Convergence Diagnosis and
  Output Analysis for MCMC}.
\newblock R News \textbf{6}(1), 7--11 (2006).
%\newblock \urlprefix\url{http://CRAN.R-project.org/doc/Rnews/Rnews_2006-1.pdf}

\bibitem{r2014}
{R Core Team}: R: A Language and Environment for Statistical Computing.
\newblock R Foundation for Statistical Computing, Vienna, Austria (2014).
%\newblock \urlprefix\url{http://www.R-project.org/}

\bibitem{van2009}
Van, H.M., Vose, D.: {A Compendium of Distributions - ebook}  (2009)

\end{thebibliography}

\end{document}